\PassOptionsToPackage{dvipsnames}{xcolor}
\documentclass[preliminary,copyright]{eptcs}
\providecommand{\event}{LICS 2018} 

\usepackage{booktabs}   
\usepackage{subcaption} 
\usepackage[T1]{fontenc}
\usepackage[utf8]{inputenc}
\usepackage{bussproofs}
\usepackage{tikz}
\usepackage{tikz-cd}
\usepackage{mathrsfs}
\usepackage{amsmath}
\usepackage{amssymb}
\usepackage{amsthm}
\usepackage{stmaryrd}

\DeclareRobustCommand{\amalg}{\mathop{\text{\fakecoprod}}}
\newcommand{\fakecoprod}{%
  \sbox0{$\prod$}%
  \smash{\raisebox{\dimexpr.9625\depth-\dp0}{\scalebox{1}[-1]{$\prod$}}}%
  \vphantom{$\prod$}%
}

\numberwithin{equation}{section}

\theoremstyle{plain}
\newtheorem{theorem}{Theorem}[section]
\newtheorem{lemma}[theorem]{Lemma}
\newtheorem{proposition}[theorem]{Proposition}

\theoremstyle{definition}
\newtheorem{definition}[theorem]{Definition}

\theoremstyle{remark}

\newenvironment{bprooftree}
  {\leavevmode\hbox\bgroup}
  {\DisplayProof\egroup}

\newcommand{\id}{\text{id}}

\newcommand{\VV}{\ensuremath{\mathbf{V}}}
\newcommand{\VVE}{\ensuremath{\mathscr{V}}}

\newcommand{\AAE}{\ensuremath{\mathscr{A}}}

\newcommand{\BBE}{\ensuremath{\mathscr{B}}}

\newcommand{\CC}{\ensuremath{\mathbf{C}}}
\newcommand{\CCE}{\ensuremath{\mathscr{C}}}

\newcommand{\DCPO}{\ensuremath{\mathbf{CPO}}}
\newcommand{\dcpoe}{\ensuremath{\mathcal{CPO}}}
\newcommand{\dcpo}{\DCPO}
\newcommand{\dcpobs}{\ensuremath{\mathbf{CPO}_{\perp!}}}
\newcommand{\dcpobse}{\ensuremath{\mathcal{CPO}_{\perp!}}}

\newcommand{\cpo}{\dcpo}

\newcommand{\M}{\ensuremath{\mathbf{M}}}
\newcommand{\ME}{\ensuremath{\mathcal{M}}}

\newcommand{\Set}{\ensuremath{\mathbf{Set}}}
\newcommand{\Cat}{\ensuremath{\mathbf{Cat}}}

\newcommand{\lift}{\text{lift}}
\newcommand{\lleft}{\text{left}}
\newcommand{\rright}{\text{right}}
\newcommand{\force}{\text{force}}
\newcommand{\ccase}{\text{case}}
\newcommand{\llet}{\text{let}}

\newcommand{\bbox}{\text{box}}
\newcommand{\Diag}{\text{Diag}}
\newcommand{\apply}{\text{apply}}
\newcommand{\ev}{\text{ev}}
\newcommand{\rec}{\text{rec}}

\DeclareMathAlphabet{\mathpzc}{OT1}{pzc}{m}{it}
\newcommand{\kay}{\ensuremath{\mathpzc{k}}}
\newcommand{\vkay}{\vec{\kay}}
\newcommand{\vell}{\vec{\ell}}

\newcommand{\lrb}[1]{{\llbracket #1 \rrbracket}}
\newcommand{\sem}[1]{{\llbracket #1 \rrbracket}}

\newcommand{\sleq}{\sqsubseteq}
\newcommand{\tleq}{\trianglelefteq}

\newcommand{\blift}{\textbf{lift}}

\newcommand{\naturalto}{%
  \mathrel{\vbox{\offinterlineskip
    \mathsurround=0pt
    \ialign{\hfil##\hfil\cr
      \normalfont\scalebox{1.2}{.}\cr
      $\longrightarrow$\cr}
}}%
}

\usetikzlibrary{decorations.pathmorphing}
\usetikzlibrary{decorations.markings}
\usetikzlibrary{decorations.pathreplacing}
\usetikzlibrary{arrows}
\usetikzlibrary{shapes}

\pgfdeclarelayer{edgelayer}
\pgfdeclarelayer{nodelayer}
\pgfsetlayers{edgelayer,nodelayer,main}

\tikzstyle{braceedge}=[decorate,decoration={brace,amplitude=10pt}]
\tikzstyle{square box}=[rectangle,fill=white,draw=black,minimum height=6mm,minimum width=6mm,yshift=0.7mm]
\tikzstyle{wire label}=[font=\footnotesize, auto,swap]

\tikzstyle{none}=[inner sep=0pt]
\tikzstyle{gn}=[circle,fill=Lime,draw=Black,line width=0.8 pt]
\tikzstyle{rn}=[circle,fill=Red,draw=Black, line width=0.8 pt]
\tikzstyle{H}=[rectangle,fill=Yellow,draw=Black]
\tikzstyle{line}=[scalar,fill=White,draw=Black]
\tikzstyle{io}=[rectangle,fill=White,draw=Black]
\tikzstyle{block}=[rectangle,fill=Orange,draw=Black]
\tikzstyle{graph}=[circle,fill=White,draw=Black]
\tikzstyle{empty}=[rectangle,fill=none,draw=none]
\tikzstyle{box}=[rectangle,fill=White,draw=Black]
\tikzstyle{dot}=[circle,fill=Black,draw=Black,inner sep=0pt,minimum size=1pt]
\tikzstyle{small dot}=[circle,fill=Black,draw=Black,inner sep=0pt,minimum size=1pt]
\tikzstyle{Dot}=[circle,fill=Black,draw=Black,inner sep=0pt,minimum size=3pt]
\tikzstyle{diam}=[rectangle,fill=Black,draw,yscale=1.2,rotate=45]
\tikzstyle{gangle}=[rectangle,fill=Lime,draw=Black]
\tikzstyle{rangle}=[rectangle,fill=Red,draw=Black]
\tikzstyle{circ}=[circle,fill=none,draw=Black,scale=1.3]
\tikzstyle{ellip}=[ellipse,fill=none,draw=Black,scale=1.3,minimum width =1.3cm]
\tikzstyle{ellip2}=[ellipse,fill=White,draw=Black,scale=1.3,minimum width =3cm]
\tikzstyle{bbox}=[rectangle,fill=Blue,draw=Blue,scale=0.6]
\tikzstyle{gg}=[shape=rectangle,fill=White,draw=Black,dashed]

\tikzstyle{white circle}=[circle,fill=none,draw=Black,scale=1]
\tikzstyle{black circle}=[circle,fill=Black,draw=Black,scale=1]
\tikzstyle{grey circle}=[circle,fill=Gray,draw=Black,scale=1]
\tikzstyle{white rectangle}=[rectangle,fill=none,draw=Black,scale=1]

\tikzstyle{nodev}=[circle,fill=none,draw=Black,scale=1]
\tikzstyle{greynode}=[circle,fill=Grey,draw=Black,scale=1]
\tikzstyle{blacknode}=[circle,fill=Black,draw=Black,scale=1]
\tikzstyle{wirev}=[circle,fill=Black,draw=Black,inner sep=0pt,minimum size=3pt]
\tikzstyle{wirevred}=[circle,fill=Red,draw=Black,inner sep=0pt,minimum size=3pt]

\tikzstyle{simple}=[-,draw=Black]
\tikzstyle{directed}=[->,draw=Black]
\tikzstyle{bdirected}=[<->,draw=Black]
\tikzstyle{bothdirs}=[bdirected,draw=Black]
\tikzstyle{bothdirsred}=[bdirected,draw=Red]
\tikzstyle{blue}=[-,draw=Blue]
\tikzstyle{redd}=[directed,draw=Red]
\tikzstyle{redu}=[-,draw=Red]
\tikzstyle{blued}=[directed,draw=Blue]
\tikzstyle{dash}=[dashed,draw=Black]
\tikzstyle{ddash}=[->,dashed,draw=Black]
\tikzstyle{dashedd}=[->,dashed]
\tikzstyle{dashedred}=[dashed,draw=Red]

\tikzstyle{dotpic}=[scale=0.5]

\tikzstyle{every picture}=[baseline=-0.25em]

\newcommand{\stikz}[2][1]{\scalebox{#1}{
\InputIfFileExists{#2}{}{\input{./tikz/#2}}
}}
\newcommand{\cstikz}[2][1]{\begin{center}\stikz[#1]{#2}\end{center}}

\title{Enriching a Linear/Non-linear Lambda Calculus:\\ A Programming Language for String Diagrams}

\author{
  Bert Lindenhovius
  \institute{Department of Computer Science\\Tulane University}
\and
  Michael Mislove
  \institute{Department of Computer Science\\Tulane University}
\and
  Vladimir Zamdzhiev
  \institute{Department of Computer Science\\Tulane University}
}
\def\titlerunning{Enriching a Linear/Non-linear Lambda Calculus: A Programming Language for String Diagrams}
\def\authorrunning{B. Lindenhovius, M. Mislove \& V. Zamdzhiev}

\begin{document}
\maketitle
\begin{abstract}
Linear/non-linear (LNL) models, as described by Benton, soundly model a LNL
term calculus and LNL logic closely related to intuitionistic linear logic.
Every such model induces a canonical enrichment that we show soundly models
a LNL lambda calculus for string diagrams, introduced by Rios and Selinger
(with primary application in quantum computing). Our abstract treatment of this
language leads to simpler concrete models compared to those presented so far.
We also extend the language with general recursion and prove soundness.
Finally, we present an adequacy result for the diagram-free fragment of the
language which corresponds to a modified version of Benton and Wadler's adjoint
calculus with recursion.

\end{abstract}
\section{Introduction}\label{sec:introduction}
In recent years string diagrams have found applications across a range 
of areas in computer science and related fields: in concurrency theory, where they 
are used to model Petri nets~\cite{petri-nets}; in systems theory, where they are used 
in a calculus of signal flow diagrams~\cite{signal-flow}; and in quantum 
computing~\cite{zw-calculus,zx-calculus}  where they represent quantum circuits and have been 
used to completely axiomatize the Clifford+T segment of quantum mechanics~\cite{zx-complete}.

But as the size of a system grows, constructing string diagram representations by hand
quickly becomes intractable, and more advanced tools are needed to accurately represent and
reason about the associated diagrams.  In fact, just generating large diagrams is
a difficult problem. One area where this has been addressed is in the development of
circuit description languages. For example, Verilog~\cite{verilog} and 
VHDL~\cite{vhdl} are popular hardware description languages that are 
used to generate very large digital circuits. More recently, the PNBml
language~\cite{pawel-programming} was developed to generate
Petri nets, and  Quipper~\cite{quipper} and QWIRE~\cite{qwire} are 
quantum programming languages (among others) that are used to 
generate (and execute) quantum circuits. 

In this paper we pursue a more abstract approach. We consider a lambda
calculus for string diagrams whose primary purpose is to generate complicated 
diagrams from simpler components. However, we do not fix a particular application domain. 
Our development only assumes that the string diagrams we are working with enjoy a 
symmetric monoidal structure. Our goal is to help lay a foundation for programming 
languages that generate string diagrams, and that support the addition of 
extensions for specific application domains along with the necessary language features. 

More generally, we believe the use of formal methods could aid us in obtaining
a better conceptual understanding of how to design languages that can be used
to construct and analyze large and complicated (families) of string diagrams.

\paragraph{Our Results.}
We study several calculi in this paper, beginning with the \emph{combined LNL}
(CLNL) calculus, which is the diagram-free fragment of our main language. The
CLNL calculus, described in Section~\ref{sec:clnl}, can be seen as a modified
version of Benton's LNL calculus,  first defined in~\cite{benton-small}. The
crucial difference is that in CLNL we allow the use of mixed contexts, so there
is only one type of judgement. This reduces the number of typing
rules, and allows us to extend the language to support the generation of
string diagrams. We also present a categorical model for our language, which 
is given by an LNL model with finite coproducts, and prove its soundness.

Next, in Section~\ref{sec:eclnl}, we describe our main language of interest,
the \emph{enriched CLNL} calculus, which we denote ECLNL. The ECLNL calculus 
adopts the syntax and operational semantics of Proto-Quipper-M, a circuit description 
language introduced by Rios and Selinger~\cite{pqm-small}, but we develop our own
categorical model. Ours is the first \emph{abstract} categorical model for the
language, which is again given by an LNL model, but endowed with an additional
\emph{enrichment} structure. The enrichment is the reason we chose to rename the
language. By design, ECLNL is an extension of the CLNL calculus that adds language
features for manipulating string diagrams. We show that our abstract
model satisfies the soundness and constructivity requirements  (see~\cite{pqm-small}, 
Remark 4.1) of Rios and Selinger's original model. As special instances of our abstract model, 
we recover the original model of Rios and Selinger, and we also present a simpler concrete
model, as well as one that is order enriched.

In Section~\ref{sec:recursion} we resolve the open problem posed by Rios and
Selinger of extending the language with general recursion. We show that all the
relevant language properties are preserved, and then we prove soundness for both
the CLNL and ECLNL calculi with recursion, after first extending our abstract models with some
additional structure. We then present concrete models for the ECLNL calculus
that support recursion and also support generating string diagrams from
\emph{any} symmetric monoidal category. We conclude the section with a
concrete model for the CLNL calculus extended with recursion that we also
prove is computationally adequate at intuitionistic types.

In Section~\ref{sec:conclusion}, we conclude the paper and discuss further
possible developments, such as adding inductive and recursive types, as well as
a treatment of dependent types. 

\paragraph{Related Work.} 
Categorical models are fundamental for our results, and the ones we present
rely on the LNL models first described by Benton in~\cite{benton-small}. 
Our work also is inspired by the language Proto-Quipper-M~\cite{pqm-small} by Rios and 
Selinger, the latest of the circuit description languages Selinger and his group have been 
developing. Our ECLNL calculus has the same syntax and operational semantics as 
Proto-Quipper-M, but there are significant differences in the denotational models. Rios 
and Selinger start with a symmetric monoidal category $\M$, then they consider a
fully faithful strong symmetric monoidal embedding of $\M$ into another category
$\overline \M$ that has some suitable categorical structure (e.g. 
$\overline \M := [\M^{\text{op}}, \Set]$), so that the category $\mathbf{Fam}(\overline \M)$ is 
symmetric monoidal closed and contains $\M$. Their model is then
given by the symmetric monoidal adjunction between $\Set$ and $\mathbf{Fam}(\overline \M),$
which allows them to distinguish
``parameter" (intuitionistic) terms and ``state" (linear) terms. They show their language is 
type safe, their semantics is sound, and they remark that it also is computationally adequate at 
observable types (there is no recursion, so all programs terminate). The semantics for our ECLNL 
calculus enjoys the same properties, but we present both an abstract model and a simpler concrete
model that doesn't involve a $\mathbf{Fam}(-)$ construction. Moreover, we also describe an extension 
with recursion, based on ideas by Benton and Wadler~\cite{benton-wadler},
and present an adequacy result for the diagram-free fragment of the language.

QWIRE~\cite{qwire} also is a language for reasoning about quantum circuits. 
QWIRE is really two languages, an intuitionistic host language and a quantum circuits language.
QWIRE led Rennela and Staton to consider a more general language
Ewire~\cite{ewire-mfps,ewire-arxiv}, which can be used to describe circuits
that are not necessarily quantum.  Ewire supports dynamic lifting, and they
prove a soundness result assuming the reduction system for the intuitionistic
language is normalizing.  They also discuss extending Ewire with conditional
branching  and inductive types over the $\otimes$- and $\oplus$-connectives
(but not $\multimap$). However, these extensions require imposing additional
structure on the diagrams, such as the existence of coproducts and
fold/unfold gates.  In our approach,
we assume only that the diagrams enjoy a symmetric monoidal structure.  In
addition, our language also supports general recursion, whereas Ewire does not.
An important similarity is that Ewire also makes use of enriched category
theory to describe the denotational model.

Aside from Ewire and Proto-Quipper-M, the other languages we mentioned
cannot generate arbitrary string diagrams, and some of them do not have a
formal denotational semantics.
\section{An alternative LNL calculus}\label{sec:clnl}
LNL models were introduced by Benton \cite{benton-small} as a means to soundly model
an interesting LNL calculus together with a corresponding logic. The goal was to understand
the relationship between intuitionistic logic and intuitionistic linear logic. In this
section, we show that LNL models also soundly model a variant of the LNL
calculus where, instead of having two distinct typing judgements (linear and
intuitionistic), there is a single type of judgement whose context is allowed to be mixed. A similar 
idea was briefly discussed by Benton in his original
paper \cite{benton-small}. The syntax and operational semantics for this language
are derived as a special case of the language of Rios and
Selinger~\cite{pqm-small}. We denote the resulting language by CLNL, which we 
call the "Combined LNL" calculus.

As with the other
calculi we consider, we begin our discussion by first describing a categorical model for CLNL. 
This makes the presentation of the language easier to follow. 
A categorical model of the CLNL calculus is given by an LNL model with finite coproducts, as
the next definition shows.

\begin{definition}[\cite{benton-small}]\label{def:lnl}
A \emph{model of the CLNL calculus} (CLNL model) is given by the following data:
a cartesian closed category (CCC) with finite coproducts $(\VV, \times, \to, 1, \amalg, \varnothing)$;
a symmetric monoidal closed category (SMCC) with finite coproducts $(\CC, \otimes, \multimap, I, +, 0)$;
and a symmetric monoidal adjunction:
  \[
  \begin{tikzcd}[ampersand replacement=\&]
    \VV \arrow[rr, bend left, "F"] \& \rotatebox{90}{$\vdash$} \& \CC \arrow[ll, bend left, "G"]
  \end{tikzcd}
  \]
We also adopt the following notation:
\begin{itemize}
\item The comonad-endofunctor is $! := F \circ G$.
\item The unit of the adjunction $F \dashv G\ $ is $\eta : \text{Id} \naturalto G \circ F.$
\item The counit of the adjunction $F \dashv G\ $ is $\epsilon :\ !  \naturalto \text{Id}.$
\end{itemize}
\end{definition}
Throughout the remainder of this section, we consider an arbitrary, but fixed, CLNL
model. The CLNL calculus, which we introduce next, is
interpreted in the category $\CC.$

The syntax of the CLNL calculus is presented in Figure~\ref{fig:syntax}. It is exactly the diagram-free fragment
of the ECLNL calculus, and because of space reasons, we only show
the typing rules for ECLNL. However, the typing rules of the CLNL calculus can
be easily derived from those for ECLNL by ignoring the $Q$ label contexts (see the (pair)
rule example below). Of course, ECLNL has some additional terms not in
CLNL, so the corresponding typing rules should be ignored as well.

Observe that the intuitionistic types are a subset of the
types of our language. Note also that there is no grammar which defines linear
types. We say that a type that is not intuitionistic is \emph{linear}.
This definition is strictly speaking not necessary, but it helps to illustrate 
some concepts. In particular, any type $A \multimap B$ is
therefore considered to be linear, even if $A$ and $B$ are intuitionistic.
The interpretation of a type $A$ is an object $\lrb{A}$ of $\CC,$ defined
by induction in the usual way (Figure~\ref{fig:semantics}).

Recall that in an LNL model with coproducts, we have:
\[I \cong F(1); \qquad 0 \cong F(\varnothing);\]
\[F(X) \otimes F(Y) \cong F(X \times Y); \qquad F(X) + F(Y) \cong F(X \amalg Y)\]
because $F$ is strong (symmetric) monoidal and also a left adjoint.
Then a simple induction argument shows:
\begin{proposition}
For every intuitionistic type $P$, there is a canonical isomorphism $\lrb P \cong F(X).$
\end{proposition}
A \emph{context} is a function from a finite set of variables to types.
We write contexts as $\Gamma = x_1 : A_1, x_2 : A_2, \ldots, x_n : A_n$, where the $x_i$ are
variables and $A_i$ are types.
Its interpretation is as usual $\lrb \Gamma = \lrb{A_1} \otimes \cdots \otimes \lrb{A_n}.$
A variable in a context is intuitionistic
(linear) if it is assigned an intuitionistic (linear) type. A context that
contains only intuitionistic variables is called an \emph{intuitionistic
context}. Note, that we do not define linear contexts, because our typing rules
refer only to contexts that either are intuitionistic  or arbitrary (mixed).

A typing judgement has the form $\Gamma \vdash m: A$, where $\Gamma$ is an
(arbitrary) context, $m$ is a term and $A$ is a type. Its interpretation
is a morphism $\lrb{\Gamma \vdash m: A}: \lrb \Gamma \to \lrb A$ in $\CC,$
defined by induction on the derivation. For the typing rules of CLNL, the label contexts $Q, Q'$, etc.\ from Figure~\ref{fig:syntax} should be ignored. For example, the (pair) rule in CLNL
becomes:
\[
\begin{bprooftree}
\AxiomC{$\Phi, \Gamma_1 \vdash m : A$}
\AxiomC{$\Phi, \Gamma_2 \vdash n : B$}
\RightLabel{(pair)} \BinaryInfC{$\Phi, \Gamma_1, \Gamma_2 \vdash \langle m, n \rangle : A \otimes B$}
\end{bprooftree}
\]
The type system enforces that a linear variable is used exactly
once, whereas a non-linear variable may be used any number of times, including
zero. Unlike Benton's LNL calculus, derivations in CLNL are in general not unique,
because intuitionistic variables may be part of an arbitrary context $\Gamma$.
For example, if $P_1$ and $P_2$ are intuitionistic types, then:
\[
\begin{bprooftree}
\AxiomC{$x:P_1 \vdash x: P_1$}
\AxiomC{$y:P_2 \vdash y: P_2$}
\RightLabel{(pair)} \BinaryInfC{$x:P_1, y:P_2 \vdash \langle x, y \rangle: P_1 \otimes P_2$}
\end{bprooftree}
\]
\mbox{}
\[
\begin{bprooftree}
\AxiomC{$x:P_1 \vdash x: P_1$}
\AxiomC{$x:P_1, y:P_2 \vdash y: P_2$}
\RightLabel{(pair)} \BinaryInfC{$x:P_1, y:P_2 \vdash \langle x, y \rangle: P_1 \otimes P_2$}
\end{bprooftree}
\]
are two different derivations of the same judgement. While this might seem to be a disadvantage,
it leads to a reduction in the number of rules, it allows a language extension that supports 
describing string diagrams (cf. Section~\ref{sec:eclnl}), 
and it allows us to easily add general recursion (cf. Section~\ref{sec:recursion}).
Moreover, the interpretation
of any two derivations of the same judgement are equal (cf. Theorem~\ref{thm:derivations}).

\begin{definition}
A morphism $f: \lrb{P_1} \to \lrb{P_2}$ is called \emph{intuitionistic}, if
\[f = \lrb{P_1} \xrightarrow{\cong} F(X) \xrightarrow{F(f')} F(Y) \xrightarrow{\cong} \lrb{P_2},\] for
some $f' \in \VV(X,Y).$
\end{definition}

\newpage

\begin{definition}
We define maps on  intuitionistic types as follows:
\begin{itemize}
\item[]\begin{enumerate}
  \item[\emph{Discard:}] $\diamond_P := \lrb{P}\xrightarrow{\cong} F(X) \xrightarrow{F(1_X)} F(1)\xrightarrow{\cong} I;$
  \item[\emph{Copy:}] $\Delta_P := \lrb{P}\xrightarrow{\cong} F(X) \xrightarrow{F(\langle \id, \id \rangle)} F(X \times X)\xrightarrow{\cong} \lrb P \otimes \lrb P;$
  \item[\emph{Lift:}] $\blift_P := \lrb{P}\xrightarrow{\cong} F(X) \xrightarrow{F(\eta_X)}\, !F(X)\xrightarrow{\cong}\ !\lrb P.$
\end{enumerate}
\end{itemize}
\end{definition}
\begin{proposition}
If $f: \lrb {P_1} \to \lrb{P_2}$ is intuitionistic, then:
\begin{itemize}
\item $\diamond_{P_2} \circ f = \diamond_{P_1};$
\item $\Delta_{P_2} \circ f = (f \otimes f) \circ \Delta_{P_1};$
\item \blift$_{P_2} \circ f =\ !f \circ$ \blift$_{P_1}.$
\end{itemize}
\end{proposition}
Because of space limitations, we are unable to provide a complete list of the operational and 
denotational semantics for the languages we discuss, so we confine ourselves to excerpts that
present some ``interesting" rules in Figures~\ref{fig:semantics} and \ref{fig:operational}. The rules for CLNL are obvious special cases of those for ECLNL (which we discuss in the next section).

The evaluation rules for CLNL can be derived from those of ECLNL
(Figure~\ref{fig:operational}) by ignoring the diagram components.
For example, the evaluation rule for (pair) is given by:
\[
\begin{bprooftree}
\AxiomC{$m \Downarrow v$}
\AxiomC{$n \Downarrow v'$}
\BinaryInfC{$\langle m, n \rangle \Downarrow \langle v, v' \rangle$}
\end{bprooftree}
\]
Similarly, the denotational interpretations of terms in CLNL can be derived
from those of ECLNL (Figure~\ref{fig:semantics}) by ignoring the $Q$ label
contexts. For example, the interpretation of
$\lrb{\Phi, \Gamma_1, \Gamma_2 \vdash \langle m,n \rangle: A \otimes B}$
is given by the composition:
\begin{align*}
      \lrb{\Phi}\otimes\lrb{\Gamma_1}\otimes\lrb{\Gamma_2}
      \xrightarrow{\Delta  \otimes \id}
      \lrb{\Phi} \otimes \lrb{\Phi}\otimes\lrb{\Gamma_1}\otimes\lrb{\Gamma_2} \xrightarrow \cong
      \lrb{\Phi}\otimes\lrb{\Gamma_1}\otimes\lrb{\Phi}\otimes\lrb{\Gamma_2}
      \xrightarrow{\lrb m \otimes \lrb n}
      \lrb A \otimes \lrb B.
\end{align*}
\begin{theorem}
Theorems~\ref{thm:derivations} -- \ref{thm:soundness} also hold true when restricted to the CLNL calculus in the obvious way.
\end{theorem}
\section{Enriching the CLNL calculus}\label{sec:eclnl}
In this section we introduce the \emph{enriched} CLNL calculus, ECLNL, whose syntax and 
operational semantics coincide with those of Proto-Quipper-M~\cite{pqm-small}. We rename the
language in order to emphasize its dependence on its abstract categorical model,   
an LNL model with an associated \emph{enrichment}. The 
categorical enrichment provides a natural framework for formulating the models we 
use, and for stating the constructivity properties (cf. Subsection~\ref{sub:construct}) that we 
want our concrete models to satisfy. 

We begin by briefly recalling the main ingredients of categories enriched over a symmetric monoidal closed category $(\VV,\otimes, \multimap, I)$: 
\begin{itemize}
	\item A $\VV$-\emph{enriched} category (briefly, a $\VV$-category) $\AAE$ consists of a collection of objects; for each pair of objects $A,B$ there is a `hom' object $\AAE(A,B)\in\VV$; for each object $A$, there is a `unit' morphism $u_A:I\to\AAE(A,A)$ in $\VV$; and given objects $A,B,C$, there is a `composition' morphism $c_{ABC}:\AAE(A,B)\otimes\AAE(B,C)\to\AAE(A,C)$ in $\VV$.
	 \item A $\VV$-\emph{functor} $F\colon \AAE\to\BBE$ between $\VV$-categories assigns to each object $A\in\AAE$ an object $FA\in\BBE$, and to each pair of objects $A,A'\in\AAE$ a $\VV$-morphism $F_{AA'}:\AAE(A,A')\to\BBE(FA,FA')$;
	\item A $\VV$-\emph{natural transformation} between $\VV$-functors $F,G:\AAE\to\BBE$ consists of $\VV$-morphisms $\alpha_A:I\to \BBE(FA,GA)$ for each $A\in\AAE$;
	\item A $\VV$-functor $F:\AAE\to\BBE$ has a right $\VV$-\emph{adjoint} $G:\BBE\to\AAE$ if there is a $\VV$-isomorphism, $\BBE(FA,B)\cong\AAE(A,GB)$ that is $\VV$-natural in both $A$ and $B$;
\end{itemize}
The $\VV$-morphisms that occur in these definitions are all subject to additional conditions expressed in terms of commuting diagrams in $\VV$; for these we refer to \cite[Chapter 6]{borceux:handbook2}, which provides a detailed exposition on enriched category theory. We denote the category of $\VV$-categories by $\VV$-$\Cat$.

The first example of a $\VV$-enriched category is the category $\VVE$ that has
the same objects as $\VV$ and whose hom objects are given by
$\VVE(A,B)=A\multimap B$. We refer to this category as the
\emph{self-enrichment} of $\VV$.  If $\AAE$ is a $\VV$-category, then the
$\VV$-\emph{copower} of an object $A\in\AAE$ by an object $X\in\VV$ is an
object $X\odot A\in\AAE$ together with an isomorphism $\AAE(X\odot
A,B)\cong\VVE(X,\AAE(A,B)),$ which is $\VV$-natural in $B$. 

Any (lax) monoidal functor $G:\CC\to\VV$ between symmetric monoidal closed categories induces a \emph{change of base} functor $G_*:\CC$-$\Cat\to\VV$-$\Cat$ assigning to each $\CC$-category $\AAE$ a $\VV$-category $G_*\AAE$ with the same objects as $\AAE$, but with hom objects given by $(G_*\AAE)(A,B)=G\AAE(A,B)$. 
In particular, if $\VV$ is locally small (which we always assume), then the
functor $\VV(I,-):\VV\to\Set$ is a monoidal functor; the
corresponding change of base functor assigns to each $\VV$-category $\AAE$ its
\emph{underlying category}, which we denote with $\mathbf{A}$, i.e., the same
letter but in boldface. We note that the underlying category of $\VVE$ is
isomorphic to $\VV$. Moreover, if the monoidal functor $G$ above has a strong
monoidal left adjoint, then the corresponding change of base functor maps
$\CC$-categories to $\VV$-categories with isomorphic underlying categories, and
$\CC$-functors to $\VV$-functors with the same underlying functors (up to the
isomorphisms between the underlying categories). If $\VV$ has all coproducts,
then $\VV(I,-)$ has a left adjoint $V:\Set\to\VV$ that is monoidal~\cite[Proposition 6.4.6]{borceux:handbook2}. Applying
the corresponding change of base functor to a locally small category equips
this category with the \emph{free} $\VV$-enrichment.

Symmetric monoidal categories can be generalized to $\VV$-\emph{symmetric monoidal} categories, where the monoidal structure is also enriched over $\VV$ \cite[\S4]{change-of-base}. It follows from
\cite[Proposition 6.3]{change-of-base} that the functor $G_*$ above maps
$\CC$-symmetric monoidal categories to $\VV$-symmetric monoidal categories. If
for each fixed $A\in\VV$, the $\VV$-functor $(- \otimes A)$ has a right
$\VV$-adjoint, denoted $(A\multimap -)$, then we call $\AAE$
a $\VV$-symmetric monoidal \emph{closed} category.
We note that the
$(- \otimes -)$ and $(- \multimap -)$ bifunctors on $\VV$ can be \emph{enriched} to
$\VV$-bifunctors on $\VVE$ (i.e., such that their underlying functors
correspond to the original functors) such that $\VVE$ becomes a $\VV$-symmetric
monoidal closed category.

Finally, if $\VV$ has finite products, a $\VV$-category $\AAE$ is said to have $\VV$-coproducts if it has an object $0$ and for each $A,B\in\AAE$ there is an object $A+B\in\AAE$ together with isomorphisms
\[1\cong \AAE(0,C),\ \ \ \AAE(A,C)\times\AAE(B,C)\cong\AAE(A+B,C),\] $\VV$-natural in $C$.

\begin{definition}\label{def:enriched-lnl}
An \emph{enriched CLNL model} is given by the following data:
\begin{enumerate}
\item A cartesian closed category $\VV$ together with its self-enrichment 
$\VVE$, such that
  $\VVE$ has finite $\VV$-coproducts; 
\item A $\VV$-symmetric monoidal closed category $\CCE$ with underlying category $\CC$ such that
    $\CCE$ has $\VV$-copowers and finite $\VV$-coproducts; 
\item A $\VV$-adjunction:
  \begin{tikzcd}
  \VVE \arrow[rr, bend left, "-\ \odot \ I"] & \rotatebox{90}{$\vdash$} &\CCE, \arrow[ll, bend left, "\CCE(I{,}-)"]
  \end{tikzcd}
  together with a CLNL model on the underlying adjunction.
\end{enumerate}
We also adopt the following notation:
$F$ and $G$ are the underlying functors of $(- \odot I)$ and $\CCE(I,-)$ respectively
and we use the same notation for the underlying CLNL model as in Definition~\ref{def:lnl}.
\end{definition}
By definition, every enriched CLNL model is a CLNL model with some additional
(enriched) structure. But as the next theorem shows, every CLNL model
induces the additional enriched structure as well.
The CCC $\VV$ can be equipped with its self-enrichment $\VVE$ in a canonical way.
The symmetric monoidal structure of the adjunction then allows us to equip the
SMCC $\CC$ with a $\VV$-enrichment by making use of the induced change-of-base
functors which stem from the adjunction. Then one can show that the now constructed
$\VV$-enriched category $\CCE$ has $\VV$-copowers and the original adjunction enriches
to a $\VV$-enriched one. We conclude:
\begin{theorem}
Every CLNL model induces an enriched CLNL model.
\end{theorem}
\begin{proof}
Combine \cite[Proposition 6.7]{eec} and \cite[Theorem 11.2]{change-of-base}.
\end{proof}
The following proposition will be useful when defining the semantics of our language.
\begin{proposition}
In every enriched CLNL model:
\begin{enumerate}
\item There is a $\VV$-natural isomorphism $G(A \multimap B) \cong \CCE(A, B);$
\item $!(A \multimap B) \cong F(\CCE(A, B)).$
\item There is a natural isomorphism $\Psi: \CC(A,B) \cong \VV(1, \CCE(A,B)).$
\end{enumerate}
\end{proposition}
\begin{proof}
  \begin{align*}
  &(1.) &&G(A \multimap B) = \CCE(I, A \multimap B) \cong \CCE(A, B);\\
  &(2.) &&\text{Apply } F \text{ to (1.)};\\
  &(3.) &&\CC(A,B) \cong\ \CC(I, A \multimap B) \cong
    \CC(F1, A\multimap B) \cong
    \VV(1, G(A\multimap B)) \cong \VV(1, \CCE(A, B)). \qedhere
  \end{align*}
\end{proof}
\subsection{The String Diagram model}
The ECLNL calculus is designed to describe string diagrams. So we first
explain exactly what kind of diagrams we have in mind. 
The morphisms of any symmetric 
monoidal category can be described using string diagrams~\cite{selinger-graphical}\footnote{The interested reader can consult~\cite{selinger-graphical} 
for more information on string diagrammatic representations of morphisms.}.
So, we choose an arbitrary 
symmetric monoidal category $\M,$ and then the string diagrams
we will be working with are exactly those that correspond to the morphisms of $\M$.

For example, if we set $\M =
\textbf{FdCStar},$ the category of finite-dimensional C*-algebras and completely
positive maps, then we can use our calculus for quantum programming. Another
interesting choice for quantum computing, in light of recent
results~\cite{zx-complete}, is setting $\M$ to be a suitable category of ZX-calculus
diagrams.
If $\M = \mathbf{PNB},$ the category of Petri Nets with Boundaries~\cite{pnb}, then our
calculus may be used to generate such Petri nets.

As with CLNL, our discussion of ECLNL begins with its categorical model. 
\begin{definition}
An \emph{ECLNL model} is given by the following data:
\begin{itemize}
\item An enriched CLNL model (Definition~\ref{def:enriched-lnl});
\item A symmetric monoidal category $(\M, \boxtimes, J)$ and a strong symmetric 
monoidal functor $E: \M \to \CC.$
\end{itemize}
\end{definition}
For the remainder of the section, we consider an arbitrary, but fixed, ECLNL model.
\subsection{Syntax and Semantics}
We first introduce new types in our syntax that correspond to the objects of
$\M.$ Using terminology introduced in~\cite{pqm-small}, where string diagrams 
are referred to as \emph{circuits}, we let  $W$ be a fixed set of \emph{wire types},
and we assume there is an interpretation $\lrb{-}_{\M} : W \to
\text{Ob}(\M).$ We use $\alpha, \beta, \ldots$ to range over the elements
of $W$. For a wire type $\alpha$, we define the interpretation of $\alpha$ in $\CC$ 
to be $\lrb \alpha = E(\lrb \alpha_M).$
The grammar for $\M$-types is given in Figure~\ref{fig:syntax}, and 
we extend $\lrb{-}_{\M}$ to $\M$-types in the obvious way.

To build more complicated string diagrams from simpler components, we need to refer
to certain wires of the component diagrams, to specify how to compose them. This
is accomplished by assigning \emph{labels} to the wires of our string diagrams, as
demonstrated in the following construction.

Let $L$ be a countably infinite set of labels. We
use letters $\ell, \kay$ to range over the elements of $L.$ A \emph{label
context} is a function from a finite subset of $L$ to $W,$
which we write as $\ell_1: \alpha_1, \ldots , \ell_n :
\alpha_n.$ We use $Q_1, Q_2, \ldots$ to refer to label contexts. To each label
context $Q = \ell_1: \alpha_1, \ldots , \ell_n : \alpha_n$, we assign an object
of $\M$ given by
$\lrb Q_{\M} := \lrb{\alpha_1}_{\M} \boxtimes \cdots \boxtimes \lrb{\alpha_n}_{\M}.$
If
$Q = \emptyset,$ then $\lrb Q_{\M} = J.$
We denote \emph{label tuples} by $\vell$ and $\vkay$; these are simply tuples of
label terms built up using the (pair) rule.

We now define the category $\M_{L}$ of \emph{labelled string diagrams}:
\begin{itemize}
  \item The objects of $\M_{L}$ are label contexts $Q$.
  \item The morphisms of $\M_{L}(Q_1, Q_2)$ are exactly
  the morphisms of $\M(\lrb{Q_1}_{\M}, \lrb{Q_2}_{\M}).$
\end{itemize}
So, by construction, $\lrb{-}_{\M}: \M_L \to \M$ is a full and faithful functor.
Observe that if $Q$ and $Q'$ are label contexts that differ only by a
renaming of labels, then $Q \cong Q'$. Moreover, for any two label contexts
$Q_1$ and $Q_2$, by renaming labels we can construct
$Q_1' \cong Q_1$ such that $Q_1'$ and $Q_2$ are
disjoint.

We equip the category $\M_L$ with the unique (up to natural isomorphism)
symmetric monoidal structure that makes $\lrb -_{\M}$ a symmetric monoidal
functor. We then have $Q \otimes Q' \cong Q \cup Q'$ for any pair of disjoint
label contexts. We use $S, D$ to range over the morphisms of $\M_L$ and we visualise
them in the following way:
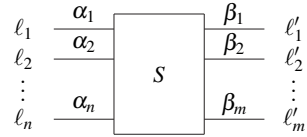
\cstikz[0.8]{labelled-string-example.tikz}
where $S: \{\ell_1: \alpha_1 , \ldots, \ell_n : \alpha_n \} \to \{\ell'_1: \beta_1 , \ldots, \ell'_m : \beta_m \} \in \M_L$
and $\lrb S_{\M} : \lrb{\alpha_1}_{\M} \boxtimes \cdots \boxtimes \lrb{\alpha_n}_{\M} \to \lrb{\beta_1}_{\M} \boxtimes \cdots \boxtimes \lrb{\beta_m}_{\M} \in \M.$ 

A label context $Q = \ell_1: \alpha_1 , \ldots, \ell_n : \alpha_n$ is
interpreted in $\CC$ as
$\lrb Q =  \lrb{\alpha_1} \otimes \cdots \otimes \lrb{\alpha_n}$
or by $\lrb Q = I$ if $Q = \emptyset.$
A labelled string diagram $S: Q \to Q'$ is interpreted in $\CC$ as the composition:
\[{\lrb{S} := \lrb Q \xrightarrow{\cong} E(\lrb Q_\M) \xrightarrow{E(\lrb S_\M)} E(\lrb{Q'}_\M)\xrightarrow{\cong} \lrb{Q'}.}\]

We also add the type
Diag$(T,U)$ to the language (see Figure~\ref{fig:syntax}); Diag$(T,U)$ 
should be thought of as the type of string diagrams with inputs
$T$ and outputs $U$, where $T$ and $U$ are $\M$-types.

The term language is extended by adding the labels and label tuples just discussed,
and the terms $\bbox_T m,\ \apply(m,n)$ and $(\vell, S,
\vell').$ The term $\bbox_T m$ should be thought of as "boxing up" an already completed diagram $m$; $\apply(m,n)$ represents the
application of the boxed diagram $m$ to the state $n$; and the term $(\vell, S, \vell')$ is a value which represents a boxed diagram.
\begin{figure*}
\begin{tabular}{l  l  l  l}
\multicolumn{4}{c}{\underline{The CLNL Calculus}}\\
  Variables & $x,y,z$ & & \\
	Types & $A, B, C$ &                                           ::= &   $0$ | $A+B$ | $I$ | $A\otimes B$ | $A \multimap B$ | $!A$ \\
	Intuitionistic types & $P, R$ &                               ::= & $0$ | $P+R$ | $I$ | $P\otimes R$ | $!A$ \\
  Variable contexts & $\Gamma$ &                                ::= & $x_1: A_1, x_2: A_2, \ldots, x_n : A_n$\\
  Intuitionistic variable contexts & $\Phi$ &                   ::= & $x_1: P_1, x_2: P_2, \ldots, x_n : P_n$\\
  Terms & $m, n, p$ & ::= & $x$ | $c$ | let $x = m$ in $n$ | $\square_C m$ | left$_{A,B} m$ | right$_{A,B} m$ | \\ 
  & & & case $m$ of $\{$left $x\to n\ |$ right $y \to p\}$ | $*$ | $m;n$ | $\langle m, n \rangle$ | \\ 
  & & & let $\langle x, y \rangle = m$ in $n$ | $\lambda x^A.m$ | $mn$ | lift $m$ | force $m$ \\
  Values & $v,w$ & ::= & $x$ | $c$ | left$_{A,B} v$ | right$_{A,B} v$ | $*$ | $\langle v, w \rangle$ | $\lambda x^A.m$ | lift $m$ \\
Term Judgements & \multicolumn{2}{l}{$\Gamma \vdash m: A$} &(typing rules below - ignore  $Q$ contexts)\\
  \multicolumn{4}{c}{\underline{The ECLNL Calculus}}\\
 \multicolumn{4}{c}{Extend the CLNL syntax with:}\\
   \color{black} Labels & \color{black} $\ell, \kay$ & & \\
  \color{black} Labelled string diagrams\phantom{ntexts}& \color{black} $S, D$ & & \\
	Types & $A,B,C$ &                                           ::= &   $\cdots$ | {\color{black}$\alpha$}  |     Diag$(T,U)$\\
Intuitionistic types & $P, R$ &   ::= &  $\cdots$ | Diag$(T,U)$ \\
  \color{black} M-types & \color{black} $T, U$ & \color{black}   ::= & \color{black} $\alpha$ | \color{black} $I$ | $T\otimes U$\\
  \color{black} Label contexts & \color{black} $Q$ & \color{black} ::= & \color{black} $\ell_1: \alpha_1, \ell_2: \alpha_2, \ldots, \ell_n : \alpha_n$\\
   Terms & $m,n,p$ & ::= & $\cdots$ |  $\ell$ | $\bbox_T m$ | $\apply(m,n)$ | $(\vell, S, \vell')$
   \phantom{  left$_{A,B} m$ | right$_{A,B} m$ | case $m$ of $\{$left $x\to n\ |$ right $y \to p\}$ |}\\
    {\color{black}Label tuples} & {\color{black}$\vell, \vkay$} & {\color{black}::=} & {\color{black}$\ell$ | $*$ | $\langle \vell, \vkay \rangle$}\\
    Values & $v,w$ & ::= &  $\cdots$ | {\color{black}$\ell$} | $(\vell,S, \vell')$\\
 Configurations & $(S, m)$ &&\\
  Term Judgements & \multicolumn{3}{l}{$\Gamma; Q \vdash m: A$ }\\
  Configuration Judgements & \multicolumn{3}{l}{$Q \vdash (S,m): A; Q'$ \qquad (cf. Definition~\ref{def:config-judgement})}\\
\multicolumn{4}{c}{\underline{The Typing Rules}}
\end{tabular}
\begin{footnotesize}
  \[
    \begin{bprooftree}
    \AxiomC{}
    \RightLabel{(var)} \UnaryInfC{$\Phi, x:A; \emptyset \vdash x: A$}
    \end{bprooftree}
    {
    \color{black} 
    \begin{bprooftree}
    \AxiomC{}
    \RightLabel{(label)} \UnaryInfC{$\Phi; \ell : \alpha \vdash \ell: \alpha$}
    \end{bprooftree}
    }
    \begin{bprooftree}
    \AxiomC{}
    \RightLabel{(const)} \UnaryInfC{$\Phi; \emptyset \vdash c: A_c$}
    \end{bprooftree}
    \qquad\qquad
    \begin{bprooftree}
    \AxiomC{$\Phi, \Gamma_1; Q_1 \vdash m : A$}
    \AxiomC{$\Phi, \Gamma_2, x : A; Q_2 \vdash n : B$}
    \RightLabel{(let)} \BinaryInfC{$\Phi, \Gamma_1, \Gamma_2; Q_1, Q_2 \vdash \llet\ x = m\ \text{in}\ n : B$}
    \end{bprooftree}
  \]

  \[
    \begin{bprooftree}
    \AxiomC{$\Gamma; Q \vdash m : 0$}
    \RightLabel{(initial)} \UnaryInfC{$\Gamma; Q \vdash \square_C m : C$}
    \end{bprooftree}
    \begin{bprooftree}
    \AxiomC{$\Gamma; Q \vdash m : A$}
    \RightLabel{(left)} \UnaryInfC{$\Gamma; Q \vdash \lleft_{A,B} m : A+B$}
    \end{bprooftree}
    \begin{bprooftree}
    \AxiomC{$\Gamma; Q \vdash m : B$}
    \RightLabel{(right)} \UnaryInfC{$\Gamma; Q \vdash \rright_{A,B} m : A+B$}
    \end{bprooftree}
    \qquad\qquad
    \begin{bprooftree}
    \AxiomC{}
    \RightLabel{(*)} \UnaryInfC{$\Phi; \emptyset \vdash *: I$}
    \end{bprooftree}
  \]

  \[
    \begin{bprooftree}
    \AxiomC{$\Phi, \Gamma_1; Q_1 \vdash m : A+B$}
    \AxiomC{$\Phi, \Gamma_2, x : A; Q_2 \vdash n : C$}
    \AxiomC{$\Phi, \Gamma_2, y : B; Q_2 \vdash p : C$}
    \RightLabel{(case)} \TrinaryInfC{$\Phi, \Gamma_1, \Gamma_2; Q_1, Q_2 \vdash \ccase\ m\ \text{of}\ \{\lleft\ x \to n\ |\ \rright\ y \to p\} : C$}
    \end{bprooftree}
    \qquad\ \ 
    \begin{bprooftree}
    \AxiomC{$\Phi, \Gamma_1; Q_1 \vdash m : I$}
    \AxiomC{$\Phi, \Gamma_2; Q_2 \vdash n : C$}
    \RightLabel{(seq)} \BinaryInfC{$\Phi, \Gamma_1, \Gamma_2; Q_1, Q_2 \vdash m;n : C$}
    \end{bprooftree}
  \]


  \[
    \begin{bprooftree}
    \AxiomC{$\Phi, \Gamma_1; Q_1 \vdash m : A$}
    \AxiomC{$\Phi, \Gamma_2; Q_2 \vdash n : B$}
    \RightLabel{(pair)} \BinaryInfC{$\Phi, \Gamma_1, \Gamma_2; Q_1, Q_2 \vdash \langle m, n \rangle : A \otimes B$}
    \end{bprooftree}
  \begin{bprooftree}
  \AxiomC{$\Phi,\Gamma_1;Q_1\vdash m : A\otimes B$}
  \AxiomC{$\Phi, \Gamma_2,x:A,y:B; Q_2 \vdash n : C$}
  \RightLabel{(let-pair)} \BinaryInfC{$\Phi, \Gamma_1, \Gamma_2; Q_1, Q_2 \vdash\llet\ \langle x,y \rangle=m\ \text{in}\ n:C$}
  \end{bprooftree}
  \]
  
  \[
    \begin{bprooftree}
    \AxiomC{$\Gamma, x: A; Q \vdash m : B$}
    \RightLabel{(abs)} \UnaryInfC{$\Gamma; Q \vdash \lambda x^A . m : A \multimap B$}
    \end{bprooftree}
    \begin{bprooftree}
    \AxiomC{$\Phi, \Gamma_1; Q_1 \vdash m : A \multimap B$}
    \AxiomC{$\Phi, \Gamma_2; Q_2 \vdash n : A$}
    \RightLabel{(app)} \BinaryInfC{$\Phi, \Gamma_1, \Gamma_2; Q_1, Q_2 \vdash mn : B$}
    \end{bprooftree}
    \begin{bprooftree}
    \AxiomC{$\Phi; \emptyset \vdash m : A$}
    \RightLabel{(lift)} \UnaryInfC{$\Phi; \emptyset \vdash \lift\ m : !A$}
    \end{bprooftree}
    \begin{bprooftree}
    \AxiomC{$\Gamma; Q \vdash m : !A$}
    \RightLabel{(force)} \UnaryInfC{$\Gamma; Q \vdash \force\ m : A$}
    \end{bprooftree}
  \]
  
  \[
  \color{black} 
  \begin{bprooftree}
  \def\ScoreOverhang{0.5pt}
  \AxiomC{$\Gamma;Q \vdash m : !(T\multimap U)$}
  \RightLabel{(box)} \UnaryInfC{$\Gamma;Q\vdash\bbox_Tm:\Diag(T,U)$}
  \end{bprooftree}
  \begin{bprooftree}
  \def\ScoreOverhang{0.5pt}
  \def\defaultHypSeparation{\hskip .1in}
  \AxiomC{$\Phi,\Gamma_1; Q_1 \vdash m : \Diag(T,U)$}
  \AxiomC{$\Phi,\Gamma_2; Q_2 \vdash n :T$}
  \RightLabel{(apply)} \BinaryInfC{$\Phi,\Gamma_1,\Gamma_2; Q_1,Q_2 \vdash \apply(m,n) : U$}
  \end{bprooftree}
  \begin{bprooftree}
  \def\ScoreOverhang{0.5pt}
  \def\defaultHypSeparation{\hskip .1in}
  \AxiomC{$\emptyset; Q\vdash\vell:T$}
  \AxiomC{$\emptyset; Q'\vdash\vell':U$}
  \AxiomC{$S \in \M_{L}(Q,Q')$}
  \RightLabel{(diag)} \TrinaryInfC{$\Phi;\emptyset\vdash(\vell,S,\vell'):\Diag(T,U)$}
  \end{bprooftree}
  \]

\end{footnotesize}
  \caption{Syntax of the CLNL and ECLNL calculi.}\label{fig:syntax}
  \end{figure*}

Users of the ECLNL programming language are not expected to write labelled 
string diagrams $S$ or terms such as $(\vell, S, \vell').$
Instead, these terms are computed by the programming language itself.
Depending on the diagram model, the language should be extended with
constants that are exposed to the user, for example, for quantum computing, a
constant $h: (\textbf{qubit} \multimap \textbf{qubit})$ could be utilised
by the user to build quantum circuits. Then the
term $\bbox_{\textbf{qubit}}\ \lift\ h$ would reduce to a term $(\ell, H, \kay)$
where $H$ is a labelled string diagram representing the Hadamard gate (where
technically each term should be part of a configuration, see below).

The term typing judgements from the previous section are now extended to include a label context as well, which is separated from
the variable context using a semicolon; the new format of a term typing judgement is $\Gamma; Q \vdash m: A.$ Its interpretation is
a morphism $\lrb \Gamma \otimes \lrb Q \to \lrb A$ in $\CC$ that is defined by induction on the derivation as shown in Figure~\ref{fig:semantics}.
\begin{figure*}
\begin{align*}
\lrb{\alpha}           &= E(\lrb \alpha_\M)\\
\lrb{0}                &= 0\\
\lrb{A + B}            &= \lrb A + \lrb B\\
\lrb{I}                &= I\\
\lrb{A \otimes B}      &= \lrb A \otimes \lrb B\\
\lrb{A \multimap B}    &= \lrb A \multimap \lrb B\\
\lrb{!A}               &= ! \lrb A\\
\lrb{\text{Diag}(T,U)} &= F(\CCE(\lrb T, \lrb U))
\end{align*}
\begin{small}
\begin{align*}
&\lrb{\Phi, \Gamma_1, \Gamma_2; Q_1, Q_2 \vdash \langle m,n \rangle: A \otimes B} :=
      \lrb{\Phi}\otimes\lrb{\Gamma_1}\otimes\lrb{\Gamma_2}\otimes\lrb{Q_1} \otimes \lrb{Q_2}
        \xrightarrow{\Delta  \otimes \id}
      \lrb{\Phi} \otimes \lrb{\Phi}\otimes\lrb{\Gamma_1}\otimes\lrb{\Gamma_2}\otimes\lrb{Q_1} \otimes \lrb{Q_2} \xrightarrow \cong\\
     &\qquad \lrb{\Phi}\otimes\lrb{\Gamma_1}\otimes\lrb{Q_1}\otimes\lrb{\Phi}\otimes\lrb{\Gamma_2}\otimes\lrb{Q_2}
      \xrightarrow{\lrb m \otimes \lrb n}
        \lrb A \otimes \lrb B\\
&\lrb{\Phi, \Gamma_1, \Gamma_2; Q_1, Q_2 \vdash\llet\ \langle x,y \rangle=m\ \text{in}\ n:C} := \lrb{\Phi}\otimes\lrb{\Gamma_1}\otimes\lrb{\Gamma_2}\otimes\lrb{Q_1} \otimes
  			\lrb{Q_2}\xrightarrow{\Delta \otimes\id}\lrb{\Phi}\otimes\lrb{\Phi}\otimes \lrb{\Gamma_1}\otimes\lrb{\Gamma_2}\otimes\lrb{Q_1} \otimes \lrb{Q_2} \xrightarrow \cong \\
        &\qquad \lrb{\Phi}\otimes\lrb{\Gamma_1}\otimes\lrb{Q_1}\otimes\lrb{\Phi}\otimes\lrb{\Gamma_2}\otimes\lrb{Q_2}
  			\xrightarrow{\lrb{m}\otimes \id}
  			\lrb{A\otimes B}\otimes \lrb{\Phi}\otimes\lrb{\Gamma_2}\otimes\lrb{Q_2} \xrightarrow \cong
       	\lrb{\Phi}\otimes\lrb{\Gamma_2}\otimes\lrb{A}\otimes\lrb{ B}\otimes \lrb{Q_2}
  			\xrightarrow{\lrb{n}}
  			\lrb{C}\\
&\lrb{\Phi; \emptyset \vdash \lift\ m :!A} := \lrb \Phi \xrightarrow{\blift} !\lrb \Phi \xrightarrow{! \lrb m} ! \lrb A\\
&\lrb{\Gamma; Q \vdash \force\ m :A} := \lrb \Gamma \otimes \lrb Q \xrightarrow{\lrb m} ! \lrb A \xrightarrow{\epsilon} \lrb A \\
&\lrb{\Gamma; Q \vdash \bbox_T m : \Diag(T,U)} := \lrb{\Gamma}\otimes\lrb{Q}\xrightarrow{\lrb{m}}\ !(\lrb{T}\multimap\lrb{U}) \xrightarrow \cong \lrb{\text{Diag}(T,U)}\\
&\lrb{\Phi, \Gamma_1, \Gamma_2; Q_1, Q_2 \vdash \apply(m,n): U} :=  \lrb{\Phi} \otimes \lrb{\Gamma_1} \otimes \lrb{\Gamma_2} \otimes \lrb{Q_1} \otimes \lrb{Q_2}
                       \xrightarrow{\Delta \otimes\id}
                       \lrb{\Phi} \otimes \lrb{\Phi} \otimes \lrb{\Gamma_1} \otimes \lrb{\Gamma_2} \otimes \lrb{Q_1} \otimes \lrb{Q_2}
                       \xrightarrow \cong\\
                       &\qquad \lrb{\Phi}\otimes\lrb{\Gamma_1}\otimes\lrb{Q_1}\otimes\lrb{\Phi}\otimes\lrb{\Gamma_2}\otimes\lrb{Q_2}
                       \xrightarrow{\lrb{m}\otimes \lrb{n}}
                       \lrb{\Diag(T,U)} \otimes \lrb{T}
                       \xrightarrow \cong\ 
                       !(\lrb{T} \multimap \lrb{U}) \otimes \lrb{T}
                       \xrightarrow{\epsilon\otimes\id}
                       (\lrb{T} \multimap \lrb{U}) \otimes \lrb{T}
                       \xrightarrow{\ev}
                       \lrb{U}
                       \\
&\lrb{\Phi; \emptyset \vdash (\vell, S, \vell') : \Diag(T,U)} := \lrb \Phi \xrightarrow{\diamond } I \xrightarrow \cong F(1) \xrightarrow{F(\Psi(\phi(\vell, S, \vell')))} \lrb{\Diag(T,U)}
\end{align*}
\end{small}
\caption{Denotational semantics of the ECLNL calculus (excerpt)}\label{fig:semantics}
\end{figure*}

In the definition of the (diag) rule in the denotational semantics, we use a
function $\phi,$ which we now explain.  From the premises of the rule, it
follows that $\lrb \vell : \lrb{Q} \to \lrb{T}$ and
$\lrb{\vell'}: \lrb{Q'} \to \lrb{U}$ are isomorphisms.
Then, $\phi(\vell, S, \vell')$ is defined to be the morphism:
\[ \phi(\vell, S, \vell') = \lrb{T} \xrightarrow{\lrb{\vell}^{-1}} \lrb{Q} \xrightarrow{\lrb S} \lrb{Q'} \xrightarrow{\lrb{\vell'}} \lrb U. \]
\begin{theorem}\label{thm:derivations}
Let $D_1$ and $D_2$ be derivations of a judgement $\Gamma; Q \vdash m: A.$
Then $\lrb{D_1} = \lrb{D_2}.$
\end{theorem}
Because of this theorem, we write $\lrb{\Gamma; Q \vdash m : A}$ instead
of $\lrb{D}.$

A \emph{configuration} is a pair $(S, m),$ where $S$ is a labelled string
diagram and $m$ is a term. Operationally, we may think of $S$ as the diagram
that has been constructed so far, and $m$ as the program which remains to be
executed.
\begin{definition}\label{def:config-judgement}
A configuration is said to be \emph{well-typed} with inputs $Q$, outputs $Q'$
and type $A$, which we write as $Q \vdash (S, m) : A; Q',$ if 
there exists $Q''$ disjoint from $Q'$, s.t. $S: Q \to Q'' \cup Q'$
is a labelled string diagram and $\emptyset; Q'' \vdash m: A.$
\end{definition}
Thus, in a well-typed configuration, the term $m$ has no free variables and its
labels correspond to a subset of the outputs of $S$. We interpret a
well-typed configuration $Q \vdash (S, m) : A; Q',$ by:
\[\lrb{(S,m)} := \lrb Q \xrightarrow{\lrb S} \lrb{Q''} \otimes \lrb{Q'} \xrightarrow{\lrb{\emptyset; Q'' \vdash m : A} \otimes \id} \lrb A \otimes \lrb{Q'}\]

The big-step semantics is defined on configurations; because
of space reasons, we only show an excerpt of the rules in
Figure~\ref{fig:operational-semantics}. The rest of the rules are standard.
A \emph{configuration value} is a configuration $(S,v)$, where $v$ is a value.
The evaluation relation $(S,m) \Downarrow (S',v)$ then relates configurations
to configuration values. Intuitively, this can be interpreted in the following way:
assuming a constructed diagram $S$, then evaluating term $m$ results in a
diagram $S'$ (obtained from $S$ by appending other subdiagrams described by $m$) 
and value $v.$ There's also an error relation $(S,m) \Downarrow$ Error which indicates 
that a run-time error occurs when we execute term $m$ from configuration $S$. There
are many such Error rules, but they are uninteresting, so we omit all but one of
them (also see Theorem~\ref{thm:errors}).

An excerpt of the operational semantics is presented in Figure~\ref{fig:operational}.
The evaluation rule for $\bbox_T m$ makes use of a function \emph{freshlabels}.
Given a $\M$-type $T$, freshlabels$(T)$ returns a pair $(Q, \vell)$ such that
$\emptyset; Q \vdash \vell: T$, where the labels in $\vell$ are fresh in
the sense that they do not occur anywhere else in the derivation. This
can always be done, and the resulting $Q$ and $\vell$ are determined uniquely,
up to a renaming of labels (which is inessential).

The evaluation rule for $\apply(m,n)$ makes use of a function \emph{append}.
Given a labelled string diagram $S''$ together with a label tuple $\vkay$
and term $(\vell, D, \vell')$, it is defined as follows.
Assuming that $\vell$ and $\vkay$ correspond exactly to the inputs of $D$
and that $\vell'$ contains exactly the outputs of $D$, then we may
construct a term $(\vkay, D', \vkay')$ which is equivalent to $(\vell, D, \vell')$
in the sense that they only differ by a renaming of labels. Moreover, we may do
so by choosing $D'$ and $\vkay'$ such that the labels in $\vkay'$ are fresh.
Then, assuming the labels in $\vkay$ correspond to a subset of the outputs of $S''$,
we may construct the labelled string diagram $S'''$ given by the composition:
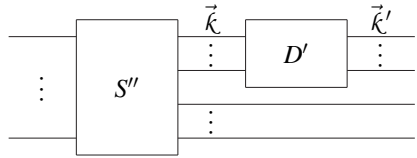
\cstikz[0.9]{append.tikz}
Finally, append$(S'', \vkay, \vell, D, \vell')$ returns the pair
$(S''', \vkay')$ if the above assumptions are met, and is undefined otherwise
(which would result in a run-time error).

\begin{figure*}
\begin{center}
\begin{bprooftree}
  \def\defaultHypSeparation{\hskip .1in}
  \def\ScoreOverhang{0.5pt}
\AxiomC{$(S, m) \Downarrow (S', v)$}
\AxiomC{$(S', n) \Downarrow (S'', v')$}
\BinaryInfC{$(S, \langle m, n \rangle) \Downarrow (S'', \langle v, v' \rangle)$}
\end{bprooftree}
\begin{bprooftree}
  \def\defaultHypSeparation{\hskip .1in}
  \def\ScoreOverhang{0.5pt}
\AxiomC{$(S, m) \Downarrow (S', \langle v, v' \rangle)$}
\AxiomC{$(S', n[v\ /\ x, v'\ /\ y]) \Downarrow (S'', w)$}
\BinaryInfC{$(S, \llet\ \langle x, y \rangle = m\ \text{in}\ n) \Downarrow (S'', w)$}
\end{bprooftree}
\end{center}
\mbox{}\\
\begin{center}
\begin{bprooftree}
  \def\defaultHypSeparation{\hskip .1in}
  \def\ScoreOverhang{0.5pt}
\AxiomC{}
\UnaryInfC{$(S, \lift\ m) \Downarrow (S, \lift\ m)$}
\end{bprooftree}
\begin{bprooftree}
  \def\defaultHypSeparation{\hskip .1in}
  \def\ScoreOverhang{0.5pt}
\AxiomC{$(S, m) \Downarrow (S', \lift\ m')$}
\AxiomC{$(S', m') \Downarrow (S'', v)$}
\BinaryInfC{$(S,\force\ m) \Downarrow (S'', v)$}
\end{bprooftree}
\end{center}
\mbox{}\\
\begin{center}
\begin{bprooftree}
  \def\defaultHypSeparation{\hskip .1in}
  \def\ScoreOverhang{0.5pt}
\AxiomC{$(S,m) \Downarrow (S', \lift\ n)$}
\AxiomC{$\text{freshlabels}(T)=(Q,\vell)$}
\AxiomC{$(\id_{Q},n\vell) \Downarrow (D, \vell')$}
\TrinaryInfC{$(S, \bbox_Tm) \Downarrow (S', (\vell,D,\vell'))$}
\end{bprooftree}
\end{center}
\mbox{}\\
\begin{center}
\begin{bprooftree}
  \def\defaultHypSeparation{\hskip .1in}
  \def\ScoreOverhang{0.5pt}
\AxiomC{$(S,m) \Downarrow (S', (\vell, D, \vell'))$}
\AxiomC{$(S',n) \Downarrow (S'', \vkay)$}
\AxiomC{append$(S'', \vkay, \vell, D, \vell') = (S''', \vkay')$}
\TrinaryInfC{$(S, \apply(m,n)) \Downarrow (S''', \vkay')$}
\end{bprooftree}
\end{center}
\mbox{}\\
\begin{center}
\begin{bprooftree}
  \def\defaultHypSeparation{\hskip .1in}
  \def\ScoreOverhang{0.5pt}
\AxiomC{$(S,m) \Downarrow (S', (\vell, D, \vell'))$}
\AxiomC{$(S',n) \Downarrow (S'', \vkay)$}
\AxiomC{append$(S'', \vkay, \vell, D, \vell')$ undefined}
\TrinaryInfC{$(S, \apply(m,n)) \Downarrow$ Error}
\end{bprooftree}
\begin{bprooftree}
  \def\defaultHypSeparation{\hskip .1in}
  \def\ScoreOverhang{0.5pt}
\AxiomC{{\color{white} $\vkay, \vell$}}
\UnaryInfC{$(S, (\vell,D,\vell')) \Downarrow (S, (\vell,D,\vell'))$}
\end{bprooftree}
\end{center}
\caption{Operational semantics of the ECLNL calculus (excerpt)}\label{fig:operational-semantics}\label{fig:operational}
\end{figure*}

\begin{theorem}[Error freeness \cite{pqm-small}]\label{thm:errors}
If $Q \vdash (S,m) :A;Q'$ then $(S,m) \not \Downarrow$ Error.
\end{theorem}

\begin{theorem}[Subject reduction \cite{pqm-small}]\label{thm:subject-reduction}
If $Q \vdash (S,m) :A;Q'$ and $(S,m) \Downarrow (S',v),$ then
$Q \vdash (S',v) :A;Q'$.
\end{theorem}
With this in place, we may now show our abstract model is sound.
We remark that our abstract model is strictly more general than the one
of Rios and Selinger (cf. Section~\ref{sec:introduction}, \emph{\textbf{Related Work}}).
\begin{theorem}(Soundness)\label{thm:soundness}
If $Q \vdash (S,m) :A;Q'$ and $(S,m) \Downarrow (S',v),$ then
$\lrb{(S,m)} = \lrb{(S',v)}$.
\end{theorem}
\subsection{A constructive property}\label{sub:construct}
If we assume, in addition, that $E:\M\to\CC$ is fully faithful, then setting
$\ME(T,U):=\CCE(ET,EU)$ for $T,U\in\M$ defines a $\VV$-enriched category $\ME$
with the same objects as $\M$, and whose underlying category is isomorphic to
$\M$. Moreover, $E$ enriches to a fully faithful $\VV$-functor $\underline E: \ME\to\CCE$. As a
consequence, our abstract model enjoys the following constructive property:
\begin{align*}
& \CC(\lrb \Phi, \lrb T \multimap \lrb U ) \cong \CC(F(X), \lrb T \multimap \lrb U) \cong\\
& \VV(X, G(\lrb T \multimap \lrb U)) \cong \VV(X, \CCE(\lrb T, \lrb U)) \cong\\
& \VV(X, \CCE( \underline E\lrb T_\M, \underline E\lrb U_\M)) = \VV(X, \ME(\lrb T_\M, \lrb U_\M))
\end{align*}
where we use the additional structure only in the last step. This means that
any well-typed term
$\Phi; \emptyset \vdash m: T \multimap U$
corresponds to a $\VV$-parametrised family of string diagrams. For example,
if $\VV=\Set$ (or $\VV=\dcpo$), then we get precisely a (Scott-continuous)
function from $X$ to $\ME(\lrb T_\M, \lrb U_\M)$ or in other words, a (Scott-continuous) family of
string diagrams from $\M.$
\subsection{Concrete Models}
The original concrete model of Rios and Selinger is now  easily recovered as an
instance of our abstract model:
\cstikz{original-model.tikz}
where $\mathbf{Fam(-)}$ is the well-known \emph{families construction}.
However, our abstract treatment of the language allows us to present a
simpler sound model:
\cstikz{simple-concrete-model.tikz}
And, an order-enriched model is given by:
\cstikz{simple-ordered-model.tikz}
where $\ME$ is the free $\dcpo$-enrichment of $\M$ (obtained by discretely
ordering its homsets) and $\dcpoe$ is the self-enrichment of $\dcpo.$
\section{The ECLNL calculus with recursion}\label{sec:recursion}
Additional structure for Benton's LNL models needed to support recursion was 
discussed by Benton and Wadler in~\cite{benton-wadler}. This structure allows 
them to model recursion in related lambda calculi, and in the LNL calculus 
(renamed the "adjoint calculus") as well. However, they present no syntax or operational semantics for 
recursion in their LNL calculus 
and instead they "\emph{\ldots omit the rather messy details}".
Here we extend both the CLNL and ECLNL calculi with recursion in a simple
way by using exactly the same additional semantic structure they use.
We conjecture the simplicity of our extension is due to our use of a single 
type of judgement that employs mixed contexts; this is the main distinguishing feature 
of our CLNL calculus compared to the LNL calculus of Benton and Wadler.
Furthermore, we 
also include a computational adequacy result for the CLNL calculus with recursion.
\subsection{Extension with recursion}

We extend the ECLNL calculus by adding the term
$\rec\ x^{!A}. m$ and we add an additional typing rule (left) and an evaluation 
rule (right) as follows:
\[
  \begin{bprooftree}
  \def\ScoreOverhang{0.5pt}
  \AxiomC{$\Phi, x: !A; \emptyset \vdash m : A$}
  \RightLabel{(rec)} \UnaryInfC{$\Phi; \emptyset \vdash \rec\ x^{!A}. m : A$}
  \end{bprooftree}
  \begin{bprooftree}
  \def\ScoreOverhang{0.5pt}
  \AxiomC{$(S, m[\lift\ \rec\ x^{!A}. m\ /\ x]) \Downarrow (S', v)$}
  \UnaryInfC{$(S, \rec\ x^{!A}. m) \Downarrow (S', v)$}
  \end{bprooftree}
\]
Notice that in the typing rule, the label contexts are empty and all free
variables in $m$ are intuitionistic.
As a special case, the CLNL calculus also can be extended with recursion:
\[
  \begin{bprooftree}
  \AxiomC{$\Phi, x: !A \vdash m : A$}
  \RightLabel{(rec)} \UnaryInfC{$\Phi \vdash \rec\ x^{!A}. m : A$}
  \end{bprooftree}
  \quad
  \begin{bprooftree}
  \AxiomC{$m[\lift\ \rec\ x^{!A}. m\ /\ x] \Downarrow v$}
  \UnaryInfC{$\rec\ x^{!A}. m \Downarrow v$}
  \end{bprooftree}
\]
In both cases, (parametrised) algebraic compactness of the $!$-endo\-functor is what is 
needed to soundly model the extension; Benton and Wadler make the same assumption.

\begin{definition}\label{def:parametric-alg-cpt}
An endofunctor $T: \mathbf C \to \mathbf C$ is \emph{algebraically compact} if
$T$ has an initial $T$-algebra $T(\Omega) \xrightarrow \omega \Omega$ for which
${\Omega \xrightarrow{\omega^{-1}} T(\Omega)}$ is a final $T$-coalgebra.
If the category $\mathbf C$ is monoidal, then
an endofunctor $T: \mathbf C \to \mathbf C$ is \emph{parametrically algebraically compact} if
the endofunctor $A \otimes T(-)$ is algebraically compact for every $A \in \mathbf C$.
\end{definition}
We note that this notion of parametrised algebraic compactness is weaker than Fiore's 
corresponding notion~\cite{fiore-thesis}, but it suffices for our purposes.
This allows us to extend both ECLNL and CLNL models with recursion in the same way.
\begin{definition}
A model of the (E)CLNL calculus with recursion is given by a model of the (E)CLNL calculus 
for which  the !-endofunctor is parametrically algebraically compact.
\end{definition}
Benton and Wadler point out that if $\mathbf C$ is symmetric monoidal closed,
then algebraic compactness of ! implies that it also is parametrically
algebraically compact. Nevertheless, we include parametric algebraic
compactness in our definition to emphasize that this is exactly what is needed 
to interpret recursion in our models.

If  $\Phi \in \mathbf C$ is an intuitionistic object, then the endofunctor
$\Phi \otimes !(-)$ is algebraically compact. Let
$\Phi \otimes !\Omega_{\Phi} \xrightarrow{\omega_{\Phi}} \Omega_{\Phi}$ be
its initial algebra and let $m: \Phi \otimes !A \to A$ be an arbitrary
morphism. We define $\gamma_{\Phi}$ and $\sigma_m$ to be the unique
anamorphism and catamorphism, respectively, such that the diagram in Figure~\ref{fig:recursion} commutes.
\begin{figure}
\cstikz{alg-compactness1.tikz}
\caption{Definition of $\sigma_m$ and $\gamma_\Phi$.}\label{fig:recursion}
\end{figure}
Using this notation, we extend the denotational semantics to interpret
recursion by adding the rule:
\[\lrb{\Phi; \emptyset \vdash \rec\ x^{!A}. m : A} := \sigma_{\lrb m} \circ \gamma_{\lrb \Phi}.\]
Observe that when $\Phi = \emptyset$, we get:
\[\lrb{\rec\ x^{!A}.m} = \lrb m \circ !\lrb{\rec\ x^{!A}. m} \circ \blift = \lrb m \circ \lrb{\lift\ \rec\ x^{!A}. m} \]
which is precisely a \emph{linear fixpoint} in the sense of Braüner~\cite{brauner}.
\begin{theorem}\label{thm:recursion}
Theorems~\ref{thm:derivations} -- \ref{thm:soundness} from the previous section remain 
true for the (E)CLNL calculus extended with recursion.
\end{theorem}
\subsection{Concrete Models}
Let $\dcpo$ be the category of cpo's (possibly without bottom) and
Scott-continuous functions, and let $\dcpobs$ be the category of \emph{pointed}
cpo's and \emph{strict} Scott-continuous functions.

We present a concrete model for an arbitrary symmetric monoidal $\M$.  Let $\ME$
be the free $\dcpo$-enrichment of $\M$. Then $\ME$ has the same objects as $\M$
and hom-cpo's $\ME(A,B)$ given by the hom-sets $\M(A,B)$ equipped with the
discrete order. $\ME$ is then a $\cpo$-symmetric monoidal category with the
same monoidal structure as $\M$.

Let $\ME_\perp$ be the free $\dcpobs$-enrichment of $\M$. Then, $\ME_\perp$ has
the same objects as $\M$ and hom-cpo's $\ME_\perp(A,B) = \ME(A,B)_\perp,$
where $(-)_\perp: \dcpoe \to \dcpobse$ is the domain-theoretic lifting functor.
$\ME_\perp$ is then a $\dcpobs$-symmetric monoidal category with the same
monoidal structure as that of $\ME$ where, in addition, $\perp_{A,B}$ satisfies the conditions of
Proposition~\ref{prop:brauner} (see Section~\ref{subsec:adequate} below).

By using the enriched Yoneda lemma together with the Day convolution monoidal
structure, we see that the enriched functor category $[\ME_\perp^{\text{op}}, \dcpobse]$
is $\dcpobs$-symmetric monoidal closed.
\begin{theorem}\label{thm:adequate-concrete}
The following data:
\cstikz{concrete-adequate-model.tikz}
is a sound model of the ECLNL calculus extended with recursion.
\end{theorem}
\begin{proof}
The subcategory inclusion
$\ME \hookrightarrow \ME_\perp$ is $\dcpo$-enriched, faithful and strong symmetric monoidal, 
as is the enriched Yoneda embedding $Y$. The $\dcpo$-copower $(- \odot I)$ is given by:
\[(- \odot I) = (- \bullet I) \circ (-)_\perp,\]
where
$(- \bullet I) : \dcpobse \to [\ME_\perp^{\text{op}}, \dcpobse]$ is the
$\dcpobs$-copower with the tensor unit (see~\cite{borceux:handbook2}). This follows because the
right adjoint and the adjunction factor through $\dcpobse$.
Parametrised algebraic compactness of the !-endofunctor follows from~\cite[pp.
161-162]{fiore-thesis}.
\end{proof}
Moreover, the concrete model enjoys a constructive property similar to the one in Subsection~\ref{sub:construct}.
Using the same argument, if $\Phi; \emptyset \vdash m : T \multimap U,$ then we obtain:
\begin{align*}
[\ME_\perp^{\text{op}}, \dcpobse](\lrb \Phi, \lrb T \multimap \lrb U )
\cong
\dcpoe(X, \ME_\perp(\lrb T_\M, \lrb U_\M))
\end{align*}
Therefore, the interpretation of $m$ corresponds to a Scott-continuous function from $X$ to
$\ME_\perp(\lrb T_\M, \lrb U_\M).$ In other words, this is a family of
\emph{string diagram computations}, in the sense that every element is either a
string diagram of $\M$ or a non-terminating computation.
\begin{theorem}\label{thm:clnl-model}
The CLNL model
\stikz{clnl-corollary.tikz},
where $U$ is the forgetful functor,
is a sound model for the CLNL calculus with recursion.
\end{theorem}
\begin{proof}
Again, parametrised algebraic compactness of the !-endofunctor follows
from~\cite[pp.  161-162]{fiore-thesis}.
\end{proof}

\subsection{Computational adequacy}\label{subsec:adequate}
In this subsection we show that computational adequacy holds at
intuitionistic types for the concrete CLNL model given in the previous subsection.

We begin by showing that in any (E)CLNL model with recursion,
the category $\CC$ is pointed, which allows us to introduce a notion of
undefinedness. Towards that end, we first introduce a slightly weaker notion,
following Braüner~\cite{brauner}.

\begin{definition}\label{def:weaklypointed}
A symmetric monoidal closed category is \emph{weakly pointed} if it is equipped with
a morphism $\perp_A:I\to A$ for each object $A$, such that for every morphism $h: A \to B$, 
we have $h\ \circ\perp_A=\perp_B.$
In this case, for each pair of objects $A$ and
$B$, there is a morphism $\perp_{A,B}= 
A \xrightarrow{\lambda_A^{-1}} I \otimes A \xrightarrow{\mathbf{uncurry}(\perp_{A\multimap B})} B.$
\end{definition}

\begin{proposition}[\cite{brauner}]\label{prop:brauner}
	Let $\mathbf A$ be a weakly pointed category. Then:
	\begin{enumerate}
		\item $f\ \circ\perp_{A,B}=\perp_{A,C}$ for each morphism $f:B\to C$;
		\item $\perp_{B,C}\circ\ f=\perp_{A,C}$ for each morphism $f:A\to B$;
		\item $\perp_{A,B}\otimes f=\perp_{A\otimes C,B\otimes D}$ for each morphism $f:C\to D$.
		\item $f \otimes \perp_{A,B} = \perp_{C\otimes A,D\otimes B}$ for each morphism $f:C\to D$.
	\end{enumerate}
\end{proposition}

\begin{lemma}
Any weakly pointed category with an initial object $0$ is pointed. Moreover,
$\perp_A = \perp_{I,A}$ and $\perp_{A,B}$ are zero morphisms. 
\end{lemma}

\begin{theorem}
For every model of the (E)CLNL calculus with recursion,
$\CC$ is a pointed category with
\[  \perp_A = I \xrightarrow{\gamma_I} \Omega_I \xrightarrow{\sigma_{\epsilon_A}} A,\]
where $\Omega_I$ is the carrier of the initial algebra for the $!$-endofunctor.
\end{theorem}
\begin{proof}
It suffices to show for any $h:A \to B$ that $h\ \circ \perp_A = \perp_B$ which
follows from the naturality of $\epsilon$ and initiality of $\sigma_\epsilon.$
\end{proof}
In particular, we have:
$\lrb{\emptyset; \emptyset \vdash \rec\ x^{!A}. \force\ x : A} = \perp_{\lrb A}.$
Thus, the interpretation of the simplest non-terminating program (of any type)
is a zero morphism, as one would expect. Naturally, we use the zero
morphisms of $\CC$ to denote undefinedness in our adequacy result.

Assume that $\CCE$ is $\dcpo$-enriched and that $\perp_{A,B}$ is least in $\CCE(A,B).$
We shall use $\bigvee_{i} a_i$ to denote the supremum of the increasing chain
$(a_i)_{i \in \mathbb N}$.
For any Scott-continuous function $K: \CCE(A,B) \to \CCE(A,B)$, let $K^0 = \perp_{A,B}$
and $K^{i+1} = K(K^i),$ for $i \in \mathbb N$. Then $\bigvee_i
K^i$ is the least fixpoint of $K$. 
Note that $K$ isn't
strict in general.
\begin{lemma}\label{lem:supremum}
Consider an (E)CLNL model with recursion, where $\VV = \dcpo$ and where $\perp_{A,B}$ is least in $\CCE(A,B),$
for all objects $A$ and $B$ (or equivalently $\CCE$ is $\dcpobs$-enriched).
Let $m: \Phi \otimes !A \to A$ be a morphism in $\CC$.
Let $K_m$ be the Scott-continuous function $K_m:\CCE(\Phi, A) \to \CCE(\Phi, A)$ given by
$K_m(f) = m \circ (\id\ \otimes !f) \circ (\id \otimes \blift) \circ \Delta.$ 
Then: 
\[\sigma_m \circ \gamma_\Phi = \bigvee_i K^i_m.\]
\end{lemma}
The significance of this lemma is that it provides an equivalent semantic definition for
the (rec) rule in terms of least fixpoints, provided we assume order-enrichment for our (E)CLNL models.

For the remainder of the section, we consider only the CLNL calculus which we
interpret in the CLNL model of Theorem~\ref{thm:clnl-model}. Therefore, in what
follows $\CC=\dcpobs$.

\begin{lemma}\label{lem:values-total}
Let $\emptyset \vdash v: P$ be a well-typed value, where $P$ is an intuitionistic type.
Then $\lrb{\emptyset \vdash v :P} \not =\ \perp.$
\end{lemma}
Next, we prove adequacy using the standard method based on 
\emph{formal approximation relations}, a notion first devised by Plotkin~\cite{plotkin-85}.
\begin{definition}
For any type $A$, let:
\begin{align*}
V_{A}    &:= \{v\ |\ v\text{ is a value and }\emptyset\vdash v:A\};\\
T_{A} &:= \{m\ |\ \emptyset \vdash m : A\}.
\end{align*}
We define two families of \emph{formal approximation relations:}
\begin{align*}
\tleq_{A} &\subseteq (\CC(I,\sem A)-\{\perp\})\times V_{A}\\
\sleq_{A} &\subseteq \CC(I , \lrb A) \times T_{A}
\end{align*}
by induction on the structure of $A$: 
\begin{itemize}
	\item[(A1)] $f\tleq_{I}*$ iff $f=\id_I$;
	\item[(A2.1)] $f\tleq_{A+B}\text{left }v$ iff $\exists f'.\ f=\text{left}\circ f'$ and $f'\tleq_{A}v$;
	\item[(A2.2)] $f\tleq_{A+B}\text{right }v$ iff $\exists f'.\ f=\text{right}\circ f'$ and $f'\tleq_{B}v$;
	\item[(A3)] $f\tleq_{A\otimes B}\langle v,w\rangle$ iff $\exists f',f'',$ such that:\\
    $f=f'\otimes f'' \circ \lambda_I^{-1}$ and $f'\tleq_{A}v$ and $f''\tleq_{B}w$;
	\item[(A4)] $f\tleq_{A \multimap B} \lambda x.\ m$ iff $\forall f' \in \CC(I, \lrb A),
    \forall v \in V_{A}:$
    \[f' \tleq_{A} v \Rightarrow \text{eval} \circ (f \otimes f')\circ \lambda_I^{-1} \sleq_{B} m[v/x];\]
	\item[(A5)] $f\tleq_{!A} \text{lift } m$ iff $f$ is an intuitionistic morphism and\\
              $\epsilon_A \circ f \sleq_{A} m;$
  \item[(B)] $f \sleq_{A} m \text{ iff } f \not = \perp\ \Rightarrow\ \exists v \in V_{A}.\ m \Downarrow v \text{ and } f \tleq_{A} v.$
\end{itemize}
\end{definition}
So, the relation $\tleq$ relates morphisms to values and $\sleq$ relates morphisms to terms.
\begin{lemma}\label{lem:approximations}
If $f \tleq_{P} v$, where $P$ is an intuitionistic type, then
  $f$ is an intuitionistic morphism.
\end{lemma}
\begin{lemma}\label{lem:scott-induction}
For any $m \in T_{A}$, the property $( - \sleq_{A} m) $ is admissible for the (pointed) cpo $\CCE(I, \lrb A)$ in the sense that Scott fixpoint induction is sound.
\end{lemma}
\begin{proof}
One has to show $\perp\ \sleq_A m,$ which is trivial, and also that ${(- \sleq_{A} m)}$ is closed under
suprema of increasing chains of morphisms, which is easily proven by induction on $A$.
\end{proof}

\begin{proposition}\label{prop:adequacy-subst}
Let $\Gamma \vdash m: A,$ where
$\Gamma = x_1: A_1, \ldots, x_n: A_n.$ Let $v_i \in V_{A_i}$ such that
$f_i \tleq_{A_i} v_i.$
If $f$ is the composition:
\begin{align*}
  &f := I \xrightarrow{\cong} I \otimes \cdots \otimes I
  \xrightarrow{f_1 \otimes \cdots \otimes f_n}
  \lrb{\Gamma}
  \xrightarrow{\lrb{\Gamma \vdash m: A}} \lrb A,
\end{align*}
then $f \sleq_{A} m[\overline v\ /\ \overline x].$
\end{proposition}
\begin{proof}
By induction on the derivation of $m.$
For the (rec) case, one should use Lemma~\ref{lem:scott-induction} and
Lemma~\ref{lem:supremum}.
\end{proof}

\begin{definition}
We shall say that a well-typed term $m$ \emph{terminates}, in symbols $m \Downarrow$,
iff there exists a value $v$, such that $m \Downarrow v.$
\end{definition}
The next theorem establishes sufficient conditions for termination at \emph{any} type.
\begin{theorem}[Termination]\label{thm:termination}
Let $\emptyset \vdash m: A\ $ be a well-typed term.
If $\ \lrb{\emptyset \vdash m: A} \not = \perp$, then $m \Downarrow.$
\end{theorem}
\begin{proof}
This is a special case of the previous proposition when $\Gamma = \emptyset$.
We get $\lrb{\emptyset \vdash m: A} \sleq_{A} m,$ and
thus $m \Downarrow$ by definition of $\sleq_{A}$.
\end{proof}

We can now finally state our adequacy result.

\begin{theorem}[Adequacy]\label{thm:adequacy}
Let $\emptyset \vdash m: P$ be a well-typed term, where $P$ is an intuitionistic type. Then:
\[m \Downarrow \text{ iff }\ \ \lrb{\emptyset \vdash m: P} \not = \perp.\]
\end{theorem}
\begin{proof}
The right-to-left direction follows from Theorem~\ref{thm:termination}.
The other direction follows from soundness and Lemma~\ref{lem:values-total}. 
\end{proof}
The model of Theorem~\ref{thm:clnl-model} was presented as an example by
Benton and Wadler~\cite{benton-wadler} for their LNL calculus extended with
recursion, however without stating an adequacy result.
We have now shown that it is computationally adequate at intuitionistic types
for our CLNL calculus. We also note that the simple proof is 
very similar to the classical proof of adequacy for PCF.
\section{Conclusion and Future Work}\label{sec:conclusion}
We considered the CLNL calculus, which is a variant of Benton's LNL calculus~\cite{benton-small},
and showed that both calculi have the same categorical models. We then showed
the CLNL calculus can be extended with recursion in a simple way while still using the same
categorical model as described by Benton and Wadler~\cite{benton-wadler}. Moreover, the CLNL
calculus also can be extended with language features that turn it into
a lambda calculus for string diagrams, which we named the ECLNL calculus (originally Proto-Quipper-M~\cite{pqm-small}). We
next identified abstract models for ECLNL by considering the categorical
enrichment of LNL models. Our abstract approach allowed us to identify concrete
models that are simpler than those previously considered, and,
moreover, it allowed us to extend the language with general recursion, thereby
solving an open problem posed by Rios and Selinger. The enrichment structure
also made it possible to easily establish the constructivity properties that one would 
expect to hold for a string diagram description language. Finally, we proved an
adequacy result for the CLNL calculus, which is the diagram-free fragment of
the ECLNL calculus.

For future work, we will consider extending ECLNL with dynamic
lifting. In quantum computing, this would allow the language to
execute quantum circuits and then use a measurement outcome to parametrize
subsequent circuit generation. Another line of future work is to consider the
introduction of inductive/recursive datatypes. Our concrete models appear to
have sufficient structure, so we believe this could be achieved in the usual
way. We will also investigate alternative proof strategies for establishing computational adequacy (at intuitionistic
types) for the ECLNL calculus.
Finally, we are interested in extending the language with dependent
types. The original model of Proto-Quipper-M was defined in terms of the
$\mathbf{Fam}(-)$ construction and has the structure of a strict indexed
symmetric monoidal category~\cite{vakar}, which suggests a potential approach
for adding type dependency.

\paragraph{Acknowledgements.}
We thank Francisco Rios and Peter Selinger for many conversations about their
work on Proto-Quipper-M, as well as on our work. We also thank Sam Staton for
raising the question of why the families construction is needed in a model of
Proto-Quipper-M (it isn't). We also thank Samson Abramsky and Mathys Rennela for valuable conversations
about this work during their recent visits to Tulane.
We also thank the Simons Institute for
the Theory of Computing  where much of the initial portion of this work took place.
This work was partially funded by the AFOSR under the MURI grant number FA9550-16-1-0082
entitled, "Semantics, Formal Reasoning, and Tool Support for Quantum
Programming".

\bibliography{refs.bib}

%

\end{document}